\documentclass[a4paper, 11pt]{article}
\usepackage[a4paper,hmargin={2.54cm,2.54cm},vmargin={3.17cm,3.17cm}]{geometry}

\usepackage{amsmath,amssymb}
\usepackage{amsthm}

\renewcommand{\Im}{\mathrm{Im}}
\renewcommand{\Re}{\mathrm{Re}}

\newtheorem{lemma}{Lemma}
\newtheorem{theorem}{Theorem}

\newtheorem{definition}{Definition}

\begin{document}

\begin{center}
\Large \textbf{Absence of traveling wave solutions of conductivity type for the Novikov-Veselov equation at zero energy}
\end{center}

\begin{center}
A.V. Kazeykina \footnote{CMAP, Ecole Polytechnique, Palaiseau, 91128, France; email: kazeykina@cmap.polytechnique.fr}
\end{center}

\textbf{Abstract.} We prove that the Novikov-Veselov equation (an analog of KdV in dimension $ 2 + 1 $) at zero energy does not have sufficiently localized soliton solutions of conductivity type.

\section{Introduction}
In this note we are concerned with the Novikov-Veselov equation at zero energy
\begin{equation}
\label{NV}
\begin{aligned}
& \partial_t v = 4 \Re ( 4 \partial_z^3 v + \partial_z( v w ) ), \\
& \partial_{ \bar z } w = - 3 \partial_z v, \quad v = \bar v, \\
& v = v( x, t ), \quad w = w( x, t ), \quad x = ( x_1, x_2 ) \in \mathbb{R}^2, \quad t \in \mathbb{R},
\end{aligned}
\end{equation}
where
\begin{equation*}
\partial_t = \frac{ \partial }{ \partial t }, \quad \partial_z = \frac{ 1 }{ 2 } \left( \frac{ \partial }{ \partial x_1 } - i \frac{ \partial }{ \partial x_2 } \right), \quad \partial_{ \bar z } = \frac{ 1 }{ 2 } \left( \frac{ \partial }{ \partial x_1 } + i \frac{ \partial }{ \partial x_2 } \right).
\end{equation*}

\begin{definition}
A pair $ ( v, w ) $ is a sufficiently localized solution of equation (\ref{NV}) if
\begin{itemize}
\item $ v, w \in C( \mathbb{R}^2 \times \mathbb{R} ) $, $ v( \cdot, t ) \in C^3( \mathbb{R}^3 ) $,
\item $ | \partial_{ x }^{ j } v( x, t ) | \leqslant \dfrac{ q( t ) }{ ( 1 + | x | )^{ 2 + \varepsilon } }, \; | j | \leqslant 3 $, for some $ \varepsilon > 0 $, $ w( x, t ) \to 0, | x | \to \infty $,
\item $ ( v, w ) $ satisfies (\ref{NV}).
\end{itemize}
\end{definition}

\begin{definition}
A solution $ (v, w) $ of (\ref{NV}) is a soliton (a traveling wave) if $ v( x, t ) = V( x - c t ) $, $ c \in \mathbb{R}^2 $.
\end{definition}

Equation (\ref{NV}) is an analog of the classic KdV equation. When $ v = v( x_1, t ) $, $ w = w( x_1, t ) $, then equation (\ref{NV}) is reduced to KdV. Besides, equation (\ref{NV}) is integrable via the scattering transform for the $ 2 $--dimensional Schr\"odinger equation
\begin{equation}
\label{schrodinger}
\begin{aligned}
& L \psi = 0, \\
L = - \Delta + v( x, t ), & \quad \Delta = 4 \partial_z \partial_{ \bar z }, \quad x \in \mathbb{R}^2.
\end{aligned}
\end{equation}

Equation (\ref{NV}) is contained implicitly in \cite{M} as an equation possessing the following representation
\begin{equation}
\frac{ \partial( L - E ) }{ \partial t } = [ L - E, A ] + B( L - E ),
\end{equation}
where $ L $ is defined in (\ref{schrodinger}), $ A $ and $ B $ are suitable differential operators of the third and zero order respectively and $ [ \cdot, \cdot ] $ denotes the commutator. In the explicit form equation (\ref{NV}) was written in \cite{NV1}, \cite{NV2}, where it was also studied in the periodic setting. For the rapidly decaying potentials the studies of equation (\ref{NV}) and the scattering problem for (\ref{schrodinger}) were carried out in \cite{BLMP}, \cite{GN} \cite{T}, \cite{LMS}. In \cite{LMS} the relation with the Calder\'on conductivity problem was discussed in detail.

\begin{definition}
A potential $ v \in L^p( \mathbb{R}^2 ) $, $ 1 < p < 2 $, is of conductivity type if $ v = \gamma^{ -1/2 } \Delta \gamma^{ 1/2 } $ for some real-valued positive $ \gamma \in L^{ \infty }( \mathbb{R}^2 ) $, such that $ \gamma \geqslant \delta_0 > 0 $ and $ \nabla \gamma^{ 1/2 } \in L^p( \mathbb{R}^2 ) $.
\end{definition}
The potentials of conductivity type arise naturally when the Calder\'on conductivity problem is studied in the setting of the boundary value problem for the $ 2 $-dimensional Schr\"odinger equation at zero energy (see \cite{Nov1}, \cite{N}, \cite{LMS}); in addition, in \cite{N} it was shown that for this type of potentials the scattering data for (\ref{schrodinger}) are well-defined everywhere.

The main result of the present note consists in the following: there are no solitons of conductivity type for equation (\ref{NV}). The proof is based on the ideas proposed in \cite{Nov2}.

This work was fulfilled in the framework of research carried out under the supervision of R.G. Novikov.

\section{Scattering data for the $ 2 $-dimensional Schr\"odinger equation at zero energy with a potential of conductivity type}
Consider the Schr\"odinger equation (\ref{schrodinger}) on the plane with the potential $ v( z ) $, $ z = x_1 + i x_2 $, satisfying
\begin{equation}
\label{v_conditions}
\begin{aligned}
& v( z ) = \overline{ v( z ) }, \quad v( z ) \in L^{ \infty }( \mathbb{C} ), \\
& | v( z ) | < q ( 1 + | z | )^{ - 2 - \varepsilon } \text{ for some } q > 0, \; \varepsilon > 0.
\end{aligned}
\end{equation}
For $ k \in \mathbb{C} $ we consider solutions $ \psi( z, k ) $ of (\ref{schrodinger}) having the following asymptotics
\begin{equation}
\label{psi_asympt}
\psi( z, k ) = e^{ i k z } \mu( z, k ), \quad \mu( z, k ) = 1 + o( 1 ), \text{ as } | z | \to \infty,
\end{equation}
i.e. Faddeev's exponentially growing solutions for the two-dimensional Schr\"odinger equation (\ref{schrodinger}) at zero energy, see \cite{F}, \cite{GN}, \cite{Nov1}.

It was shown that if $ v $ satisfies (\ref{v_conditions}) and is of conductivity type, then $ \forall k \in \mathbb{C} \backslash 0 $ there exists a unique continuous solution of (\ref{schrodinger}) satisfying (\ref{psi_asympt}) (see \cite{N}). Thus the scattering data $ b $ for the potential $ v $ of conductivity type are well-defined and continuous:
\begin{equation}
b( k ) = \iint\limits_{ \mathbb{C} }  e^{ i( k y + \bar k \bar y ) } v( y ) \mu( y, k ) d \Re y d \Im y, \quad k \in \mathbb{C} \backslash 0.
\end{equation}
In addition (see \cite{N}), the function $ \mu( z, k ) $ from (\ref{psi_asympt}) satisfies the following $ \bar \partial $-equation
\begin{equation}
\label{d_bar}
\frac{ \partial \mu ( z, k ) }{ \partial \bar k } = \frac{ 1 }{ 4 \pi \bar k } e^{ - i ( k z + \bar k \bar z ) } b( k ) \overline{ \mu( z, k ) }, \quad z \in \mathbb{C}, \quad k \in \mathbb{C} \backslash 0
\end{equation}
and the following limit properties:
\begin{gather}
\label{mu_limit}
\mu( z, k ) \to 1, \text{ as } | k | \to \infty, \\
\label{mu_bounded}
\mu( z, k ) \text{ is bounded in the neighborhood of } k = 0.
\end{gather}

The following lemma describes the scattering data corresponding to a shifted potential.
\begin{lemma}
Let $ v( z ) $ be a potential satisfying (\ref{v_conditions}) with the scattering data $ b( k ) $. The scattering data $ b_y( k ) $ for the potential $ v_y( z ) = v( z - y ) $ are related to $ b( k ) $ by the following formula
\begin{equation}
\label{b_y}
b_y( k ) =  e^{ i ( k y + \bar k \bar y ) } b( k ), \quad k \in \mathbb{C} \backslash 0, \quad y \in \mathbb{C}.
\end{equation}
\end{lemma}
\begin{proof}
We note that $ \psi( z - y, k ) $ satisfies (\ref{schrodinger}) with $ v_y( z ) $ and has the asymptotics $ \psi( z - y, k ) = e^{ i k ( z - y ) }( 1 + o( 1 ) ) $ as $ | z | \to \infty $. Thus $ \psi_y( z, k ) = e^{ i k y } \psi( z - y, k ) $ and $ \mu_y( z, k ) = \mu( z - y, k ) $. Finally, we have
\begin{multline*}
b_y( k ) = \iint\limits_{ \mathbb{C} } e^{ i( k \zeta + \bar k \bar \zeta ) } v_y( \zeta ) \mu_y( \zeta, k ) d \Re \zeta d \Im \zeta = \\
= \iint\limits_{ \mathbb{C} } e^{ i( k \zeta + \bar k \bar \zeta ) } v( \zeta - y ) \mu( \zeta - y, k ) d \Re \zeta d \Im \zeta = e^{ i ( k y + \bar k \bar y ) } b( k ).
\end{multline*}
\end{proof}

As for the time dynamics of the scattering data, in \cite{BLMP}, \cite{GN} it was shown that if the solution $ ( v, w ) $ of (\ref{NV}) exists and the scattering data for this solution are well-defined, then the time evolution of these scattering data is described as follows:
\begin{equation}
\label{time_evolution}
b( k, t ) = e^{ i ( k^3 + \bar k^3 ) t } b( k, 0 ), \quad k \in \mathbb{C} \backslash 0, \quad t \in \mathbb{R}.
\end{equation}

\section{Absence of solitons of conductivity type}
\begin{theorem}
Let $ (v, w) $ be a sufficiently localized traveling wave solution of (\ref{NV}) of conductivity type. Then $ v \equiv 0 $, $ w \equiv 0 $.
\end{theorem}
\begin{proof}[Scheme of proof]
From (\ref{b_y}), (\ref{time_evolution}), continuity of $ b( k ) $ on $ \mathbb{C} \backslash 0 $ and the fact that the functions $ k $, $ \bar k $, $ k^3 $, $ \bar k^3 $, $ 1 $ are linearly independent in the neighborhood of any point, it follows that $ b \equiv 0 $. Equation (\ref{d_bar}) implies that in this case the function $ \mu( z, k ) $ is holomorphic on $ k $, $ k \in \mathbb{C} \backslash 0 $. Using properties (\ref{mu_limit}) and (\ref{mu_bounded}) we apply Liouville theorem to obtain that $ \mu \equiv 1 $. Then $ \psi( z, k ) = e^{ i k z } $ and from (\ref{schrodinger}) it follows that $ v \equiv 0 $.
\end{proof}


\begin{thebibliography}{90}
\bibitem[BLMP]{BLMP} Boiti M., Leon J.J.P., Manna M., Pempinelli F.: On a spectral transform of a KdV-like equation related to the Schrodinger operator in the plane. Inverse Problems. 3, 25--36 (1987)
\bibitem[F]{F} Faddeev L.D. Growing solutions of the Schr\"odinger equation. Dokl. Akad. Nauk SSSR. 165, 514-517 (1965); translation in Sov. Phys. Dokl. 10, 1033-1035 (1966)
\bibitem[GN]{GN} Grinevich P.G., Novikov S.P.: Two-dimensional ``inverse scattering problem'' for negative energies and generalized-analytic functions. I. Energies below the ground state. Funct. Anal. Appl. 22(1), 19--27 (1988)
\bibitem[LMS]{LMS} Lassas M., Mueller J.L., Siltanen S.: Mapping properties of the nonlinear Fourier transform in dimension two. Communications in Partial Differential Equations. 32, 591--610 (2007)
\bibitem[M]{M} Manakov S.V.: The inverse scattering method and two-dimensional evolution equations. Uspekhi Mat. Nauk. 31(5), 245--246 (1976) (in Russian)
\bibitem[N]{N} Nachman A.I.: Global uniqueness for a two-dimensional inverse boundary value problem. Annals of Mathematics. 143, 71--96 (1995)
\bibitem[Nov1]{Nov1} Novikov R.G.: Multidimensional inverse spectral problem for the equation $ - \Delta \psi + ( v( x ) - E u( x ) ) \psi =0 $. Funkt. Anal. i Pril. 22(4), 11-22 (1988); translation in Funct. Anal. Appl. 22, 263-272 (1988)
\bibitem[Nov2]{Nov2} Novikov R.G.: Absence of exponentially localized solitons for the Novikov--Veselov equation at positive energy. Physics Letters A. 375, 1233--1235 (2011)
\bibitem[NV1]{NV1} Novikov S.P., Veselov A.P.: Finite-zone, two-dimensional, potential Schr\"odinger operators. Explicit formula and evolutions equations. Dokl. Akad. Nauk SSSR. 279, 20--24 (1984), translation in Sov. Math. Dokl. 30, 588-591 (1984)
\bibitem[NV2]{NV2} Novikov S.P., Veselov A.P.: Finite-zone, two-dimensional Schr\"odinger operators. Potential operators. Dokl. Akad. Nauk SSSR. 279, 784--788 (1984), translation in Sov. Math. Dokl. 30, 705--708 (1984)
\bibitem[T]{T} Tsai T.-Y. The Schr\"odinger operator in the plane. Inverse Problems. 9, 763--787 (1993)
\end{thebibliography}
\end{document}